\documentclass[]{interact}
\usepackage{epstopdf}
\usepackage[caption=false]{subfig}
\usepackage[numbers,sort&compress]{natbib}
\bibpunct[, ]{[}{]}{,}{n}{,}{,}
\renewcommand\bibfont{\fontsize{10}{12}\selectfont}
\makeatletter
\def\NAT@def@citea{\def\@citea{\NAT@separator}}
\makeatother

\theoremstyle{plain}
\newtheorem{theorem}{Theorem}[section]
\newtheorem{lemma}[theorem]{Lemma}
\newtheorem{corollary}[theorem]{Corollary}
\newtheorem{proposition}[theorem]{Proposition}
\newtheorem{assumption}{Assumption}
\usepackage{arydshln} 
\theoremstyle{definition}

\theoremstyle{remark}
\newtheorem{remark}{Remark}

\newtheorem{problem}{Problem}
\usepackage{lscape} 
\allowdisplaybreaks

\begin{document}

\articletype{RESEARCH ARTICLE}

\title{Sparse Feedback Controller: From Open-loop Solution to Closed-loop Realization}

\author{
  \name{Zhicheng Zhang$^{\dagger}$ and Yasumasa Fujisaki$^{\dagger}$}
  \thanks{Emails: \{zhicheng-zhang;fujisaki\}@ist.osaka-u.ac.jp}
  \affil{$^{\dagger}$Graduate School of Information Science and Technology, 
    Osaka University, Japan}
}

\maketitle

\begin{abstract}
In this paper, we explore the discrete time sparse feedback control 
for a linear invariant system, 
where the proposed optimal feedback controller enjoys input sparsity 
by using a dynamic linear compensator, i.e., 
the components of feedback control signal having the smallest possible nonzero values.
The resulting augmented dynamics ensures closed-loop stability, 
which infers sparse feedback controller from open-loop solution to closed-loop realization. 
In particular, 
we show that the implemented sparse feedback (closed-loop) control solution is 
equivalent to the original sparse (open-loop) control solution under a specified basis. 
We then extend the dynamic compensator to a feedforward tracking control problem.
Finally, numerical examples demonstrate the effectiveness of proposed control approach.
\end{abstract}


\section{Introduction}

Sparse control is closely related to sparse optimization \cite{donoho2006compressed}, 
which penalizes sparsity on controller so as to schedule resource-aware allocation.
In control design, 
the sparsity-promoting idea has thrived in various directions to distributed control 
\cite{jovanovic2016controller,matni2016regularization,furieri2020sparsity}, 
tracking control \cite{zhang2021linear}, 
and predictive control \cite{nagahara2016discrete,darup2021fast,aguilera2017quadratic}.
When the sparsity is imposed on the control structure (i.e., structured sparsity) 
whose controller depends on a \emph{static} state 
feedback \cite{lin2013design,polyak2013lmi,babazadeh2016sparsity}, 
then the sparse control is recast into a distributed control, 
which attempts to \emph{reduce the number of communication links} 
\cite{jovanovic2016controller} in networked control system. 
Another appealing alternative is to penalize sparsity on control signal 
(i.e., input sparsity) to implement a sparse control 
that \emph{maximizes the time duration over which the control value is exactly zero} 
\cite{nagahara2015maximum}. 
In this paper, we are interested in the generation of control signals, 
which emphasizes on the latter sparse control, namely, 
reducing control effort or fuel consumption, also known as ``{$\ell_{1}$ optimal control}'' 
or ``{maximum hands-off control}'' \citep{nagahara2016discrete,bakolas2019computation}. 

The optimal control design of sparse signals is of great significance 
for control system and directly affects the dynamic process of the system. 
Many practical systems of interest are dependent on a feedback mechanism 
to achieve closed-loop stability. 
However, closed-loop realization is more challenging to \emph{optimal control problem} 
because \emph{determining the feedback gain (matrix)} is a non-trivial task 
\cite{blanchini2003relatively}.
Indeed, in closed-loop sparse control scenarios, 
almost all existing literature has been focused on discussing ``structured sparsity'' 
\cite{jovanovic2016controller,matni2016regularization,furieri2020sparsity} 
by optimizing linear quadratic state feedback cost 
\citep{lin2013design,polyak2013lmi,babazadeh2016sparsity} 
rather than pursuing our expected sparse control inputs 
(i.e., $\ell_{1}$ optimal control) \cite{nagahara2016discrete,caponigro2017sparse}. 
Furthermore, most successful stories on sparse control taking $\ell_{1}$ cost of 
discrete (resp., $\mathcal{L}_{1}$ cost of continuous) systems have been extensively treated in \emph{open-loop solutions} 
\cite{nagahara2015maximum,nagahara2016discrete,bakolas2019computation,rao2017sparsity}.
These recent advances motivate us to study 
the \emph{closed-loop $\ell_{1}$ optimal control} problem.

Although ``{real-time control}'' bridges the gap 
between the open-loop and closed-loop solutions, 
schemes such as self-triggered sparse control \cite[Sec.~VI]{nagahara2015maximum} and 
sparse predictive control 
\cite{nagahara2016discrete,iwata2020realization,aguilera2017quadratic,darup2021fast} can, 
and often do, emerge feedback solutions.
On the other hand, these iterative feedback algorithms perform \emph{online optimization}, 
leading to the computationally burden, 
especially when the decision variable is high dimension.
Neither exploring sparsity on the structure of feedback gain matrix or 
exploiting real-time control, we immediately promotes sparsity on the control inputs 
with a closed-loop response, called \emph{sparse feedback control realization}. 
Inspired by seminal works \cite{blanchini2003relatively,blanchini2015youla}, 
a relatively optimal control technique paves the way towards the open-loop solution to 
the closed-loop solution by means of linear implementation.

In this paper, we focus on \emph{closed-loop realization} 
for sparse optimal control design for a discrete linear time invariant system, 
where the state feedback controller is linear dynamic, as well as enjoys input sparsity. 
We become aware of the result \cite{bykov2018sparse} that induces a desired sparse input 
by taking a \emph{row-sparsity} on the \emph{static} state feedback gain, 
which implies a structured sparsity on channels. 
In contrast, we here leverage a \emph{dynamic state feedback} controller and 
keep the standard optimal control framework. 
We believe that such a direct sparse optimization of control signal is 
more convenient to synthesizing sparse feedback control. 
One of its benefits is that it relies solely on \emph{offline optimization}. 
Thus, the proposed controller can often offer a significant computational 
saving and control effort minimization. 
Encouraged by simple theoretical and numerical results 
by a position paper \cite{zhang2023iscssparse} of the present authors, 
we extend the results to a more comprehensive version, 
including problem setup, proofs, tracking control, and simulations.
The main contributions of this article are summarized as follows.
\begin{itemize}
  \item This paper gives feedback realization of sparse optimal control 
    via a dynamic linear compensator.
    In other words, the sparse feedback controller is derived from open-loop solutions 
    of sparse optimal control.
    Besides, we propose the computationally tractable analytical and 
    explicit feedback solution for the sparse control problem 
    (Problems \ref{prob:sparse-matrix-optim} and \ref{prob:sparse-feedback-realization}).
  \item We provide the stability, optimaility, and sparsity for the closed-loop 
    augmented system, and display that the designed sparse feedback control is 
    essentially a deadbeat control (Theorem~\ref{them:sparse-feedback-realziation} 
    and Corollary~\ref{coro:deadbeat}).
  \item In particular, we show that an equivalence for the open-loop sparse control and 
    the closed-loop sparse control under a specified basis (Corollary~\ref{coro:equiva}).
  \item Furthermore, we demonstrate that the sparse feedback control can in fact be 
    extended successfully for tracking problem (Lemma \ref{lemma:tracking-steady-state} 
    and Proposition \ref{prop:feedforward-tracking}).
\end{itemize}

This paper is organized as follows. 
In Section~\ref{sec:problem-formula}, 
we introduce the problem formulation and the basic preliminaries. 
Section~\ref{sec:sfb} gives the main result for stabilizing sparse feedback control 
using a dynamic linear compensator, which can be divided into two steps, 
that is, sparse optimization and feedback realization, respectively. 
Section~\ref{sec:sfb-tracking} extends the result to a tracking problem 
by devising a dynamic tracking controller. 
The numerical examples are illustrated in Section~\ref{sec:simulations}. 
Section~\ref{sec:conclusion} concludes this paper.

{\bf{Notation}.} 
Throughout this paper, let $\mathbb{R}$, $\mathbb{R}^{n}$, 
and $\mathbb{R}^{n\times{m}}$ denote the sets of real numbers, $n$ dimension of real vectors, 
and $n\times{m}$ size real matrices, respectively.
We use $I_{n}$ (or $0_{n}$) to denote the identity (or zero) matrix of size $n\times{n}$, 
$0_{n\times{m}}$ to denote the zero matrix of size $n\times{m}$; 
and for brevity, we sometimes abbreviate $I$ (or $0$) to represent the identify 
(or zero) matrix with appropriate size. 
Let ${\mathbf{1}}_{n}=[1~\cdots~1]^{\top}\in\mathbb{R}^{n}$ be the all one vector, 
and ${\rm{e}}_{i}\in\mathbb{R}^{q}$ stands for an $q$-tuple basis vector 
whose all entries equal to 0, except the $i$th, which is $1$.
Given a vector $x=[x_{1}~x_{2}~\cdots~x_{n}]^{\top}\in\mathbb{R}^{n}$, 
we define the $\ell_{1}$ and $\ell_{2}$ norms, respectively, 
by $\|x\|_{1}=\sum_{i=1}^{n}|x_{i}|$, $\|x\|_{2}=\sqrt{\sum_{i=1}^{n}|x_{i}|^{2}}$.
Similarity, the $\ell_{1}$ norm of a matrix $X\in\mathbb{R}^{n\times{m}}$ is defined 
by $\|X\|_{1}=\sum_{i=1}^{n}\sum_{j=1}^{m}|x_{ij}|$.

\section{Problem formulation}\label{sec:problem-formula}

\subsection{Review of sparse optimal control}\label{subsec:sys-descr}

Consider a discrete linear time invariant (LTI) system described by
\begin{equation}
\begin{aligned}
  x(t+1) &= Ax(t) + Bu(t), & x(0) &= x_{0}, \\
  y(t) &= Cx(t) + Du(t),
  \label{eq:DT-LTI}
\end{aligned}
\end{equation}
where $u(t)\in{{\mathbb{R}}}^{m}$ is the control input with $m\leq{n}$, 
$x(t)\in{{\mathbb{R}}}^{n}$ is the state with an initial value $x_{0}$, 
$y(t)\in{{\mathbb{R}}}^{p}$ is the output,
and $A$, $B$, $C$, and $D$ are real constant matrices of appropriate sizes.
Throughout this paper, we assume that the pair $(A,B)$ is \emph{reachable}.

In this paper, we are interested in sparse optimal control problem, 
and the control objective is to seek a control sequence 
$\{u(t)\}_{t=0}^{N-1}$ 
such that it drives the resultant state $x(t)$ from an initial state $x(0)=x_{0}$ 
to the origin in a finite $N$ steps (i.e., $x(N)=0$) with minimum or sparse control effort. 

As indicated in \cite{donoho2006compressed},
an exact sparsity is achieved by 
penalizing an $\ell_{0}$ ``quasi-norm'' on decision variables. 
However, computing the $\ell_{0}$ norm precisely is challenging 
due to its non-convex and non-smooth nature, often resulting in an NP-hard problem. 
It suggests replacing the $\ell_{0}$ cost with convex relaxation 
using an $\ell_{1}$ norm, which still generates the sparse solution. 
In particular, the restricted isometry property reveals an equivalence 
between $\ell_{0}$ (sparse) optimal control and $\ell_{1}$ optimal control for 
discrete systems \cite{nagahara2016discrete}.
In this context, we shift the idea from compressed sensing to sparse control problem, 
which seeks an ``{open-loop}'' $\ell_{1}$ optimal control action $u^{\ast}$ 
in control system \cite{nagahara2016discrete,bakolas2019computation} defined as 
\begin{align}
  u^{\ast} &= \arg\min_{u\in{\mathcal{U}}}~ \|u\|_{1}
  = \arg\min_{u\in{\mathcal{U}}}~\sum_{t=0}^{N-1} \|u(t)\|_{1},
  \label{eq:open-loop-l1-optimal}
\end{align}
where $u=\big[u^{\top}(0)~u^{\top}(1)~\cdots~u^{\top}(N-1)\big]^{\top}\in{{\mathbb{R}}}^{mN}$, 
and $\|u\|_{1}$ indicates the $\ell_{1}$ norm of input vector $u$ 
that sums the absolute values of its elements.
Meanwhile, a feasible control set is described by 
$\mathcal{U} = \{u\in{{\mathbb{R}}}^{mN}:\Phi_{N}{u}=-A^{N}x_{0}\}$,
in which $\Phi_{N}=[A^{N-1}B~\cdots~AB~B]\in{{\mathbb{R}}}^{n\times{mN}}$ is an $N$ step 
reachability matrix and satisfies full row rank, i.e., $\mathrm{rank}(\Phi_{N})=n$. 
Occasionally, the state and input constraints are necessarily taken into account, 
for instance, $y(t) \in \mathcal{Y}$, 
where ${\mathcal{Y}}$ is a convex and closed set. 
Besides, the horizon $N$ should be sufficiently long so that 
the admissible set of $u$ is \emph{nonempty}.

\subsection{Dynamic linear compensator}
Before proceeding with the ``feedback realization'', 
a compensator ${\mathcal{K}}$ which we demand to design is a \emph{dynamic} 
and \emph{linear} state feedback controller, depending on the evolution
\begin{align}
  \mathcal{K}: \quad
  \begin{aligned}
    z(t+1) &= Fz(t) + Gx(t), & z(0) &= 0, \\
    u(t) &= Hz(t) + Kx(t),
    \label{eq:dynamic-state-feedback}
  \end{aligned}
\end{align}
where $z(t)$ is the state of the compensator 
and $F$, $G$, $H$, and $K$ are real matrices of compatible sizes.
Note that we set the initial state $z(0)$ of the compensator as \emph{zero} (i.e., $z(0)=0$). 

The advantages of such a dynamic compensator ${\mathcal{K}}$ 
of \eqref{eq:dynamic-state-feedback} are threefold. 
First, it brings a \emph{linear dynamic} fashion to the sparse feedback control realization, 
which allows for computationally tractable compensator gain matrices. 
Second, it promotes \emph{input/temporal} sparsity 
\cite{nagahara2015maximum,nagahara2016discrete,bakolas2019computation}, 
rather than structured/spatial sparsity for the controller 
\cite{lin2013design,polyak2013lmi,matni2016regularization,jovanovic2016controller,babazadeh2016sparsity,furieri2020sparsity}. 
Lastly, it ensures \emph{internal stability} for the closed-loop system.

Notice that
the requirement $z(0) = 0$ is not overly restrictive in our context. 
In fact, in the following section, we further impose $z(N) = 0$, along with $x(N) = 0$, as part of the sparse control implementation. 
This means that we consider the closed-loop system through \emph{deadbeat control}. 
In this case, once the stable closed-loop reaches the zero state within a finite time, 
the condition $z(0) = 0$ is automatically fulfilled 
whenever we have another $x(0) \neq 0$ due to a new disturbance.

\section{Sparse feedback control realization}\label{sec:sfb}

In this section, we focus on optimal sparse feedback control 
synthesis \emph{from open-loop solution to closed-loop realization}. 
Determining the explicit matrices of $F$, $G$, $H$, $K$ 
for dynamic compensator \eqref{eq:dynamic-state-feedback} is of 
primary interest in designing sparse feedback controller.

Let us formally state the \emph{constrained sparse optimal control problem} with \emph{a general initial condition} $x_{0} \in \mathcal{X}_{0}$, where $\mathcal{X}_{0}\doteq\{{\rm{e}}_{1},{\rm{e}}_{2},\ldots,{\rm{e}}_{n}\}$ is used to \emph{generate all $n$ possible input-state trajectories}, and ${\rm{e}}_{i} \in \mathbb{R}^{n}$ represents the standard basis vector, e.g., ${\rm{e}}_{1} = [1~0~\cdots~0]^\top \in {\mathbb{R}}^n$.
The problem is as follows:
Find $F$, $G$, $H$, and $K$ of \eqref{eq:dynamic-state-feedback} such that 
\begin{enumerate}
  \item[(i)] the dynamic compensator \eqref{eq:dynamic-state-feedback} 
    stabilizes the plant \eqref{eq:DT-LTI} and 
  \item[(ii)] for any $x_{0} \in \mathcal{X}_{0}$
    with $z(0) = 0$, the controller \eqref{eq:dynamic-state-feedback} 
    generates an input sequence $\{u(t)\}_{t=0}^{N-1}$, which minimizes $\sum_{t=0}^{N-1}\|u(t)\|_1$ 
    subject to the terminal constraints $x(N) = 0$ and $z(N) = 0$ for a positive integer $N$, 
    as well as the state and input constraints
    \begin{align}
      y(t) &\in \mathcal{Y}, &
      \mathcal{Y} &= \{ y \in \mathbb{R}^{p} : -s \leq y(t) \leq s \}, 
      \label{constrs:state-input}
    \end{align}
    where $s\in\mathbb{R}^{p}$ is a given positive vector. 
\end{enumerate}

\begin{remark}
 The input sparsity can be easily performed by minimizing the convex $\ell_1$ norm 
instead of the non-convex $\ell_0$ norm, as seen in the previous section. 
Moreover, although we select only a subset of the initial states of the plant, i.e., 
$x_{0} \in\mathcal{X}_{0}$, 
which is sufficient for our purposes.
In fact, suppose that the given constrained sparse control problem is solved. 
Since the resultant closed-loop system 
composed of \eqref{eq:DT-LTI} and \eqref{eq:dynamic-state-feedback} is \emph{linear}, 
it means that for any $x_0 \in \mathbb{R}^n$ with $z(0) = 0$, 
the controller \eqref{eq:dynamic-state-feedback} generates a \emph{linear combination} 
of the input sequences corresponding to $x_{0} = {\rm{e}}_{1}$, $x_{0} = {\rm{e}}_{2}$, $\ldots$, 
$x_{0} = {\rm{e}}_{n}$, thereby achieving sparsity and satisfying $x(N) = 0$ and $z(N) = 0$.
\end{remark}

\subsection{Closed-loop augmented system}

In the celebrated works \cite{blanchini2003relatively,blanchini2015youla}, 
the authors performed the linear implementation built from the relatively optimal technique. 
This oracle suggests us to investigate an augmented closed-loop system 
composed of the discrete dynamics \eqref{eq:DT-LTI} and 
the dynamic compensator \eqref{eq:dynamic-state-feedback} of the form 
\begin{align}
\begin{aligned}
  \psi(t+1) &= (\mathcal{A} + \mathcal{B}\mathcal{K}) \psi(t), &
  \psi(0) &= \psi_0, \\
  y(t) &= (\mathcal{C} + \mathcal{D}\mathcal{K}) \psi(t),
  \label{eq:aug-sys-xz}
  \end{aligned}
\end{align}
where the corresponding state and the gain matrices are given by
\begin{align*}
  \psi(t) &= \begin{bmatrix} x(t) \\ z(t) \end{bmatrix}, &
  \psi_0 &= \begin{bmatrix} x_0 \\ 0 \end{bmatrix} \\
  {\mathcal{A}} &= \begin{bmatrix} A & 0 \\ 0 & 0 \end{bmatrix}, &
  {\mathcal{B}} &= \begin{bmatrix} B & 0 \\ 0 & I \end{bmatrix}, &
  {\mathcal{K}} &= \begin{bmatrix} K & H \\ G & F \end{bmatrix}, \\
  {\mathcal{C}} &= \begin{bmatrix} C & 0 \end{bmatrix}, &
  {\mathcal{D}} &= \begin{bmatrix} D & 0 \end{bmatrix}.
\end{align*}

To represent the closed-loop system dynamics in a compact way, 
we accordingly introduce a stable matrix $P$, 
which is an $N$-Jordan block associated with $0$ eigenvalue, defined by
\begin{align}
  {P} &= \begin{bmatrix} 0 & 0 \\ I_{N-1} & 0 \end{bmatrix} \in {\mathbb{R}}^{N\times{N}}.
  \label{eq:stable-P}
\end{align}

Based on the previous problem setup, 
we formulate the constrained sparse optimal control problem 
as the following sparse optimization.

\begin{problem}[Sparse Optimization]\label{prob:sparse-matrix-optim} 
Find the matrices $X \in {\mathbb{R}}^{n \times nN}$ and $U \in {\mathbb{R}}^{m \times nN}$ 
such that the obtained $U$ is sparse, 
which amounts to solve a constrained $\ell_{1}$ norm input matrix optimization
\begin{align*}
  \min_{X, U} \quad & \| U \|_{1} \\
  \textrm{s.t.} \quad 
  & AX + BU = X (P \otimes I_{n}), \\
  & I_{n} = X ({\rm{e}}_1 \otimes I_{n}), \\
  & {\mathrm{abs}}(CX+DU) \leq s({\bf{1}}_{n}\otimes{\bf{1}}_{N})^{\top}, 
\end{align*}
where ${\rm{e}}_1=[1~0~\cdots~0]^{\top}\in\mathbb{R}^{N}$ and $\mathrm{abs}(\cdot)$ returns the absolute value of each element in a matrix.
\end{problem}

It is mentioned that Problem \ref{prob:sparse-matrix-optim} is a convex optimization, 
and hence the solution is computationally tractable by means of the off-the-shelf packages, 
such as {\texttt{CVX}} \cite{grant2008cvx} 
or {\texttt{YALMIP}} \cite{lofberg2004yalmip} in {\text{MATLAB}}.

Once the open-loop optimal solution $(X, U)$ of Problem \ref{prob:sparse-matrix-optim} 
is attained, we then proceed the second step, 
that is to say, we move on to tackling the following sparse feedback realization problem. 

\begin{problem}[Feedback Realization]\label{prob:sparse-feedback-realization}
Based on the solution $(X,U)$ of Problem~\ref{prob:sparse-matrix-optim},
solve a linear equation 
\begin{align}
  \begin{bmatrix} K & H \\ G & F \end{bmatrix} \begin{bmatrix} X \\ Z \end{bmatrix}
  = \begin{bmatrix} U \\ V \end{bmatrix} 
  \label{eq:feedb-K}
\end{align}
with respect to $(K,H,G,F)$ and determine the compensator's gain matrices, where 
\begin{align}
  Z &= \begin{bmatrix} 0_{n(N-1)\times{n}} & I_{n(N-1)} \end{bmatrix}, &
  V &= Z (P \otimes I_{n}).
  \label{eq:matrices-Z-V}
\end{align}
\end{problem}

\subsection{Stability analysis}

The approach to sparse feedback control design makes use of 
the above discussed sparse optimization (Problem \ref{prob:sparse-matrix-optim}) and 
feedback realization (Problem \ref{prob:sparse-feedback-realization}), 
where the dynamic compensator ensures the internally stability of 
the closed-loop augmented system \eqref{eq:aug-sys-xz}. 
The result is summarized in the following theorem and corollary.

\begin{theorem}[Sparse Feedback Control Realization]\label{them:sparse-feedback-realziation}
Suppose that Problem \ref{prob:sparse-matrix-optim} has the minimizer $(X, U)$. 
Then the equation \eqref{eq:feedb-K} has the unique solution $(K,H,G,F)$ and 
the resulting compensator \eqref{eq:dynamic-state-feedback} with $z(0)=0$ generates 
the input sequence $u(t) = U ({\rm{e}}_{t+1} \otimes x_0)$, $t = 0, 1, \ldots, N-1$, 
for ${x}_{0}\in\mathcal{X}_{0}$,
which drives the plant state $x(t)$ from $x(0) = x_0$ to $x(N) = 0$ 
under the output constraint \eqref{constrs:state-input}.
Furthermore, the closed-loop system \eqref{eq:aug-sys-xz} is internally stable.
\end{theorem}

\begin{proof}
We first describe the matrices $X$ and $U$ as 
\begin{align*}
  X &= \begin{bmatrix} X_0 & X_1 & \cdots & X_{N-1} \end{bmatrix}, &
  U &= \begin{bmatrix} U_0 & U_1 & \cdots & U_{N-1} \end{bmatrix}, 
\end{align*}
where $X_t \in {\mathbb{R}}^{n \times n}$, $U_t \in {\mathbb{R}}^{m \times n}$, 
and $t= 0, 1, \ldots, N-1$.
Notice that checking the second constraint of Problem~\ref{prob:sparse-matrix-optim} 
gives rise to the result $X_0 = I_n$.
With this fact, it admits that the matrix $\Psi$ is \emph{non-singular}. 
In other words, we claim that
\begin{align}
  \det \Psi &\neq 0, &
  \Psi &= \begin{bmatrix} X \\ Z \end{bmatrix}
  = \left[\begin{array}{c;{1pt/1pt}c}
      I_{n} & X_{1}~\cdots~X_{N-1}\\  \hdashline[1pt/1pt]
      0_{n(N-1)\times{n}} & I_{n(N-1)}
  \end{array}\right].
  \label{eq:Psi-inv}
\end{align}
Since the matrix $\Psi$ is invertible, it follows that 
the equation \eqref{eq:feedb-K} has the unique solution $(K,H,G,F)$ 
associated with the dynamic compensator \eqref{eq:dynamic-state-feedback}.
Meanwhile, the first constraint of Problem~\ref{prob:sparse-matrix-optim} 
with \eqref{eq:feedb-K} and \eqref{eq:matrices-Z-V} asserts that
\begin{align*}
  ({\mathcal{A} + \mathcal{BK}}) \Psi
  &= \begin{bmatrix} A & 0 \\ 0 & 0 \end{bmatrix}
    \begin{bmatrix} X \\ Z \end{bmatrix}
    + \begin{bmatrix} B & 0 \\ 0 & I \end{bmatrix}
  \begin{bmatrix} U \\ V \end{bmatrix} \\
  &= \begin{bmatrix} AX + BU \\ V \end{bmatrix} \\
  &= \begin{bmatrix} X \\ Z \end{bmatrix} (P \otimes I_{n}) \\
  &= \Psi (P \otimes I_{n}).
\end{align*}
This implies that the closed-loop matrix 
${\mathcal{A}+\mathcal{BK}}=\Psi(P\otimes{I_{n}})\Psi^{-1}$, 
so that it is similar to a nilpotent matrix $(P\otimes{I_{n}})$, 
hence, the closed-loop system \eqref{eq:aug-sys-xz} is internally stable 
and the zero terminal state $x(N) = 0$ is achieved 
for any initial state $\psi(0)$ of the system.

Moreover, we see that the sequences $x(t) = X ({\rm{e}}_{t+1} \otimes x_0)=X_{t}x_{0}$, 
$u(t) = U ({\rm{e}}_{t+1} \otimes x_0)=U_{t}x_{0}$, 
and $z(t)=Z({\rm{e}}_{t+1}\otimes{x}_{0})=Z_{t}x_{0}$ are indeed generated 
by the system \eqref{eq:DT-LTI} with $x(0)=x_{0}$ and 
the controller \eqref{eq:dynamic-state-feedback} with $z(0)=0$. 
As a matter of fact, apparently $x(0)=X(\mathrm{e}_{1}\otimes{x}_{0})=X_{0}x_{0}=x_{0}$ 
and $z(0)=Z(\mathrm{e}_{1}\otimes{x}_{0})=0$. 
Furthermore, based on the fact 
that $(P \otimes I_{n})( {\rm{e}}_{t+1} \otimes x_0 )= {\rm{e}}_{t+2} \otimes x_0$, we have
\begin{align}
  Ax(t) + Bu(t) 
  &= (AX + BU) ({\rm{e}}_{t+1} \otimes {x}_{0}) \notag \\
  &= X (P \otimes {I}_{n}) ({\rm{e}}_{t+1} \otimes {x}_{0}) \notag\\
  &= X ({\rm{e}}_{t+2} \otimes {x}_{0}) \notag \\
  &= x(t+1),       
  \label{eq:X-et}
\end{align}
\begin{align}
  Fz(t) + Gx(t) 
  &= (FZ + GX) ({\rm{e}}_{t+1} \otimes {x}_{0}) \notag \\
  &= Z (P \otimes {I}_{n}) ({\rm{e}}_{t+1} \otimes {x}_{0}) \notag \\
  &= Z ({\rm{e}}_{t+2} \otimes {x}_{0}) \notag \\
  &= z(t+1),       
  \label{eq:Z-et}
\end{align}
\begin{align}
  Hz(t) + Kx(t) 
  &= (HZ + KX) ({\rm{e}}_{t+1} \otimes {x}_{0}) \notag \\
  &= U ({\rm{e}}_{t+1} \otimes {x}_{0}) \notag \\
  &= u(t).        
  \label{eq:U-et}
\end{align}

We next consider the third constraint of Problem \ref{prob:sparse-matrix-optim}, 
whose validity can be inspected by assessing the inequality
\begin{align*}
  -{\mathrm{abs}}(CX + DU) \leq {CX + DU} &\leq {\mathrm{abs}}(CX + DU)
\end{align*}
holds true. 
Therefore, for a given positive vector $s\in\mathbb{R}^{p}$ 
and ${x}_{0}\in\mathcal{X}_{0}$, we have
\begin{align}
  ({CX + DU}) ({\rm{e}}_{t+1} \otimes {x}_{0}) 
  &\leq {\mathrm{abs}}(CX + DU)({\rm{e}}_{t+1} \otimes {x}_{0}) \notag \\
  &\leq s ({\bf{1}}_{n} \otimes {\bf{1}}_{N})^{\top} ({\rm{e}}_{t+1} \otimes {x}_{0}) = s.
  \label{eq:Y-et-ineq}
\end{align}
Notice that $Cx(t)+Du(t) = (CX + DU)({\rm{e}}_{t+1}\otimes{x}_{0})=y(t)$, 
then the output constraint of the form $\{-s\leq y(t) \leq s\}$ 
in \eqref{constrs:state-input} is verified.

According to the above arguments, 
we claim that the sequences \eqref{eq:X-et}, \eqref{eq:Z-et}, \eqref{eq:U-et}, 
and \eqref{eq:Y-et-ineq} indeed satisfy the input-state trajectories of
 LTI dynamics \eqref{eq:DT-LTI} under output constraint \eqref{constrs:state-input}, 
which proves the theorem for realizing sparse feedback control.
\end{proof}

\begin{remark}[Offline vs. Online]
It is clear that realizing sparse feedback control in 
Theorem \ref{them:sparse-feedback-realziation} is based on \emph{offline computation}, 
and hence the computational complexity is low. 
Compared with sparse predictive control 
\cite{nagahara2016discrete,aguilera2017quadratic,iwata2020realization}, 
a real-time feedback iterations is employed 
to ensure closed-loop dynamics and \emph{online optimization} is repeatedly performed 
as a feedback controller to calculate sparse solutions. 
Beyond all doubt, predictive feedback control naturally leads to computational burden 
when the sizes of controlled system is high (e.g., the curse of dimensionality), 
even for using a fast ADMM (alternating direction method of multipliers) 
algorithm \cite{nagahara2016discrete,darup2021fast}.
\end{remark}

Based on the proposed Theorem \ref{them:sparse-feedback-realziation}, 
we can directly give a corollary that establishes the connection between the open-loop 
sparse optimal control solution and the closed-loop sparse optimal control solution.

\begin{corollary}[Equivalence]\label{coro:equiva}
Suppose that Problem \ref{prob:sparse-matrix-optim} has the minimizer $(X, U)$. 
Let $u_{\mathcal{K}}^{\ast}$ be optimal sparse feedback control 
(i.e., closed-loop $\ell_{1}$ optimal control) solution 
using a dynamic linear compensator $\mathcal{K}$ \eqref{eq:dynamic-state-feedback}, 
and $u^{\ast}$ be open-loop $\ell_{1}$ optimal control solution $u^{\ast}$ 
of program \eqref{eq:open-loop-l1-optimal} with 
output constraint \eqref{constrs:state-input}, respectively.
 Then, for ${x}_{0}\in\mathcal{X}_{0}$, it holds that
\begin{align}
\begin{aligned}
             {u^{\ast}} = u_{\mathcal{K}}^{\ast} = Hz + Kx^{\ast}.
\end{aligned}
\end{align}
\end{corollary}

\begin{corollary}[Deadbeat Control]\label{coro:deadbeat}
Suppose that Problems \ref{prob:sparse-matrix-optim} 
and \ref{prob:sparse-feedback-realization} have solved,
then the implemented sparse feedback controller $u_{\mathcal{K}}^{\ast}=Hz+Kx^{\ast}$ 
(i.e., closed-loop $\ell_{1}$ optimal control) of discrete LTI plant \eqref{eq:DT-LTI} 
is essentially an $N$-step \emph{deadbeat controller}. 
\end{corollary}

\begin{remark}
Regarding the deadbeat control, since the designed compensator $\mathcal{K}$ 
brings the state $x(t)$ to the origin in $N$ steps (satisfying $x(N)=0$), 
which places all of the eigenvalues of the augmented closed-loop 
system matrix $\mathcal{A}+\mathcal{B}\mathcal{K}$
at the origin in the complex plane. 
\end{remark}

\section{Extension: Tracking problem}\label{sec:sfb-tracking}

In this section, we extend the above result of sparse feedback control 
to tracking control problem \citep[Chapter~8]{lewis2012optimal}. 

We start by giving a step-type reference signal $r(t)\in{\mathbb{R}}^{p}$ as
\begin{align}
  r(t) &= \left\{ \begin{aligned} r_{-}, && t<0 \\ r_{+}, && t\geq 0, \end{aligned} \right.
   \label{eq:reference-dynamics}
\end{align}
where $r_{-}\in\mathbb{R}^{p}$ and $r_{+}\in\mathbb{R}^{p}$ are constant vectors.
The purpose of tracking problem is to design a dynamic tracking compensator 
such that the performance output tracks a reference input with zero steady-state error 
by using additional feedforward gains. 
For this reason, we define the tracking error by $e(t)=y(t) - r(t)$, 
where $y(t)$ is a performance output signal stated in LTI plant \eqref{eq:DT-LTI}.
Meanwhile, we make the following assumption before giving an effective tracking controller. 
\begin{assumption}\label{assum:tracking-prob}
For a tracking problem, assume that the performance output signal $y(t)\in\mathbb{R}^{p}$ 
and the control signal $u(t)\in\mathbb{R}^{m}$ in LTI dynamics \eqref{eq:DT-LTI} be of 
the same size (i.e., $p=m$) and take the matrix $D=0_{m}$.
\end{assumption}

It is known that the performance output $y(t)\in\mathbb{R}^{m}$ can track any reference 
signal $r(t)\in\mathbb{R}^{m}$ of \eqref{eq:reference-dynamics} in the steady-state if 
\begin{align}
  {\mathrm{rank}}~\begin{bmatrix} I - A & B \\ C & 0 \end{bmatrix} &= n + m.
  \label{eq:rank-tracking}
\end{align}
As already reported in Section~\ref{sec:sfb}, 
an analogous dynamic tracking compensator $\mathcal{K}_{r}$ can be applied to 
the discrete LTI plant \eqref{eq:DT-LTI} by adding a prescribed reference 
input $r(t)\in\mathbb{R}^{m}$ to the control actuator \eqref{eq:dynamic-state-feedback}.
To this end, a dynamic tracking compensator $\mathcal{K}_{r}$ for the plant can be designed as
\begin{align}
  \mathcal{K}_{r}: \quad
  \begin{aligned}
    z_{r}(t+1) &= Fz_{r}(t) + Gx(t) + Lr(t), & z_{r}(0) &= 0, \\
    u(t) &= Hz_{r}(t) + Kx(t) + Mr(t),
    \label{eq:dynamic-tracking-comp}
  \end{aligned}
\end{align}
where $L$ and $M$ represent the feedforward gain matrices with suitable sizes, 
and $r(t)\in\mathbb{R}^{m}$ is a specific reference input \eqref{eq:reference-dynamics}. 
Notice that here the initial value $z_{r}(0)$ of tracking compensator ${\mathcal{K}}_{r}$ 
is set to zero (i.e., $z_{r}(0)=0$).
Figure~\ref{fig:feedforward_tracking} shows the closed-loop system composed of 
the plant \eqref{eq:DT-LTI} and the controller \eqref{eq:dynamic-tracking-comp}.

\begin{figure}
  \centering
 \includegraphics[width=0.8\linewidth]{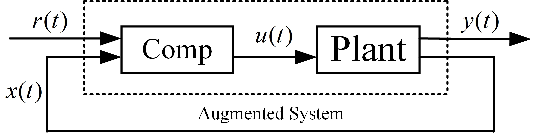}
  \caption{Feedforward tracking control system: $r(t)$ is reference signal; 
    $y(t)$ is the performance output which must track a specified reference input $r(t)$;
     ``${\text{Comp}}$'' represents a dynamic tracking compensator 
    \eqref{eq:dynamic-tracking-comp} applied to a discrete LTI plant \eqref{eq:DT-LTI}.}
  \label{fig:feedforward_tracking}
\end{figure}

For a preferable reference input tracking, 
we employ \emph{the difference or variation of the control inputs} 
$\sum_{t=0}^{N-1}\| u(t+1) - u(t)\|_{1}$ as the performance index, 
which is referred to as \emph{minimum attention control} 
\cite{brockett1997minimum,nagahara2020approach,lee2021existence}.
We slightly relax the constraints in the previous sections 
by removing the state and input constraints \eqref{constrs:state-input}.
As a result, we formulate the following \emph{tracking (minimum attention) control problem}
that we aim to solve here:
Find $F$, $G$, $H$, and $K$ of \eqref{eq:dynamic-tracking-comp} such that
\begin{enumerate}
  \item[(i)] the dynamic tracking compensator \eqref{eq:dynamic-tracking-comp} 
    stabilizes the plant \eqref{eq:DT-LTI} and 
  \item[(ii)] for any $x_{0} \in \mathcal{X}_{0}$ 
    with $z_{r}(0) = 0$ and $r(t) \equiv 0$, 
    the controller \eqref{eq:dynamic-tracking-comp} 
    generates an input sequence $\{u(t)\}_{t=0}^{N-1}$,
    which minimizes $\sum_{t=0}^{N-1} \| u(t+1) - u(t) \|_{1}$ 
    subject to $x(N) = 0$ and $z(N) = 0$ for a positive $N$.
\end{enumerate}
Then, determine $L$ and $M$ of \eqref{eq:dynamic-tracking-comp} such that 
the steady state gain of the closed-loop system from $r(t)$ to $y(t)$ 
is the identity and that from $r(t)$ to $z_r(t)$ is zero.

\begin{remark}
When we have a solution to the above problem, 
we see that $y(t)$ tracks $r(t)$ without steady state error 
owing to the selected steady state gain.
We also observe that $u(t)$ achieves minimum attention 
for any $x_0\in\mathbb{R}^{n}$ due to linearity of the system.
Moreover, since $z(N) = 0$ is achieved in the steady state for any $r_{+}$, 
the condition $z(0) = 0$ is automatically satisfied
whenever we have another $r_{+}$ as a new reference signal. 
\end{remark}

The closed-loop behavior can be described in an augmented description
\begin{align}
\begin{aligned}
  \psi_{r}(t+1) &= ({\mathcal{A} + \mathcal{BK}}) \psi_{r}(t) + \mathcal{M}_{r}r(t), &
  \psi_{r}(0) &= \psi_{r0}, \\
  y(t) &=\mathcal{C}\psi_{r}(t),
  \label{eq:aug-xz-track}
\end{aligned}
\end{align}
where 
\begin{align*}
  \psi_{r}(t) &= \begin{bmatrix} x(t) \\ z_{r}(t) \end{bmatrix}, &
  \mathcal{M}_{r} &= \begin{bmatrix} BM \\ L \end{bmatrix}, 
\end{align*}
and the other matrices $\mathcal{A}$, $\mathcal{B}$, $\mathcal{K}$, and 
$\mathcal{C}$ have been defined for \eqref{eq:aug-sys-xz}.

Based on the problem setup above, 
we first consider the following minimum attention control problem.

\begin{problem}[Minimum Attention Control]\label{prob:track-attention}
Find the matrices $X \in {\mathbb{R}}^{n \times nN}$ and $U \in {\mathbb{R}}^{m \times nN}$ 
such that the obtained $U$ solves a minimum attention control problem
\begin{align*}
  \min_{X, U} \quad & \|U(P\otimes{I}_{n})-U\|_{1} \\
  \textrm{s.t.} \quad 
  & AX + BU = X(P \otimes I_{n})\\
  & I_{n} = X ({\rm{e}}_1 \otimes I_{n}).
\end{align*}
\end{problem}

Looking for the solution $(X,U)$ of Problem \ref{prob:track-attention} is 
always accessible because, the above problem is a convex program. 
Then, with $(K,H,G,F)$ of \eqref{eq:feedb-K} and \eqref{eq:matrices-Z-V},
the resultant closed-loop system \eqref{eq:aug-xz-track} always assures internal stability, 
that is, the condition 
${\mathcal{A}+\mathcal{B}\mathcal{K}}=\Psi(P\otimes{I_{n}})\Psi^{-1}$ holds, 
as discussed in Section~\ref{sec:sfb}.

We next deal with tracking problem.
Due to the fact that the closed-loop augmented system \eqref{eq:aug-xz-track} is 
internally stable with matrices $(K,H,G,F)$,
it admits a unique steady-state $\psi_{\infty} = [x_{\infty}^{\top}~z_{\infty}^{\top}]^{\top}$ 
for a desired reference input $r_{+} = \lim_{t\to\infty}r(t)$.

More precisely, we have the following matrix equation
\begin{align}
  \begin{bmatrix} \psi_{\infty} \\ y_{\infty} \end{bmatrix}
  &= \begin{bmatrix} 
    \mathcal{A} + \mathcal{B}\mathcal{K} & \mathcal{M}_{r} \\ 
    \mathcal{C} & 0 \end{bmatrix}
    \begin{bmatrix} \psi_{\infty} \\ r_{+} \end{bmatrix}.
  \label{eq:steady-aug-psi-y-track}
\end{align}
If $y_{\infty}=Cx_{\infty}={r}_{+}$ for any reference $r_{+}$, 
the output $y(t)$ tracks reference $r(t)$ with no steady-state tracking error.
If $z_{\infty}=0$ for any $r_{+}$, we can have $z(0)=0$ 
whenever the reference signal changes again after a steady-state is achieved.

In what follows, we are going to achieve tracking error elimination with $z_{\infty}=0$ 
by assigning feedforward tracking gains $M$ and $L$, leading to the following lemma.

\begin{lemma}[Steady-State Tracking]\label{lemma:tracking-steady-state}
Supposed that Assumption \ref{assum:tracking-prob} and 
rank condition \eqref{eq:rank-tracking} of steady-state tracking hold, 
and Problem \ref{prob:track-attention} has the miniminer $(X,U)$. 
Then, the priori gain matrices $(K,H,G,F)$ of 
compensator $\mathcal{K}_{r}$ \eqref{eq:dynamic-tracking-comp} can be uniquely 
determined by feedback realization \eqref{eq:feedb-K} and \eqref{eq:matrices-Z-V}. 
In addition, if the following conditions hold 
\begin{enumerate}
  \item[(i)] ${\mathrm{det}}~(I - (A + BK)) \neq 0$,
  \item[(ii)] ${\mathrm{det}}~(I - F) \neq 0$,
\end{enumerate}
then the feedforward tracking gain matrices $(M,L)$ in 
compensator ${\mathcal{K}_{r}}$ \eqref{eq:dynamic-tracking-comp} can be derived by
\begin{align}
  M &= \big( C\big(I - (A + BK)\big)^{-1}B \big)^{-1}, \label{eq:feedforward-ref-gain-M} \\
  L &= -G \big( I - (A + BK) \big)^{-1} BM. \label{eq:feedforward-ref-gain-L}
\end{align}
Based on the obtained matrices $(K,H,G,F,M,L)$,
for all initial state $x_{0}\in\mathbb{R}^{n}$ and any reference $r_{+}\in\mathbb{R}^{m}$, 
there exist the unique steady-state values $(x_{\infty},z_{\infty})$ 
such that $y_{\infty}=r_{+}$ and $z_{\infty}=0$ are achieved in the steady-state.
\end{lemma}

\begin{proof}
Since the gain matrices $(K,H,G,F)$ can previously be calculated by \eqref{eq:feedb-K} 
and \eqref{eq:matrices-Z-V}, the augmented closed-loop system \eqref{eq:aug-xz-track} is 
internally stable, as reported in Theorem~\ref{them:sparse-feedback-realziation}.
We here reformulate the augmented system \eqref{eq:steady-aug-psi-y-track} as
\begin{align}
  \begin{bmatrix} x_{\infty} \\ z_{\infty} \end{bmatrix}
  &= \begin{bmatrix} A + BK & BH \\ G & F \end{bmatrix}
    \begin{bmatrix} x_{\infty} \\ z_{\infty}\end{bmatrix}
    + \begin{bmatrix} BM \\ L \end{bmatrix} {r}_{+}.
  \label{eq:steady-aug-xz-track}
\end{align}
Suppose that the zero steady-state $z_{\infty} = 0$, we then have 
\begin{align*}
  x_{\infty} &= (A + BK) x_{\infty} + BM r_{+}.
\end{align*}
This implies that, if the matrix $I-(A+BK)$ is invertible,
\begin{align*}
  x_{\infty} = (I - (A + BK))^{-1} BM r_{+}.
\end{align*}
Therefore, we see the result $y_{\infty} = Cx_{\infty} = r_{+}$ 
if $M$ is selected as \eqref{eq:feedforward-ref-gain-M}.

On the other hand, the steady-state of the tracking compensator $z_{\infty}$ 
in \eqref{eq:steady-aug-xz-track} satisfies
\begin{align*}
  z_{\infty} = G x_{\infty} + F z_{\infty} + L r_{+}.
\end{align*}
Consequently, we derive 
\begin{align*}
  z_{\infty} &= (I - F)^{-1} (G x_{\infty} + L r_{+}) \notag \\
  &= (I - F)^{-1} \big( G(I - (A + BK))^{-1} BM + L \big) r_{+}
\end{align*}
when the matrix $(I-F)$ is invertible. Thus, we see that $z_{\infty}=0$ if
the feedforward gain matrix $L$ meets \eqref{eq:feedforward-ref-gain-L}.
Therefore, the feedforward gains \eqref{eq:feedforward-ref-gain-M} 
and \eqref{eq:feedforward-ref-gain-L} make error cancellation for achieving tracking.
\end{proof}

\begin{remark}[Illustrative conditions]
Checking determinant conditions (i) and (ii) in Lemma \ref{lemma:tracking-steady-state} 
can be equivalently formulated 
with optimal solution $X$ to Problem \ref{prob:track-attention}.
In fact, we have 
\begin{align*}
  I - (A+BK) &= I_{n} - X_{1}, \\
  I - F &= I_{n(N-1)} - \left[\begin{array}{c;{1pt/1pt}c}
      {-X}_{1}~\cdots~{-X}_{N-2} & {-X}_{N-1} \\ \hdashline[1pt/1pt]
      I_{n(N-2)}  & 0_{n(N-2)\times{n}}
    \end{array}\right]
\end{align*}
with $\mathcal{A} + \mathcal{B}\mathcal{K} = \Psi (P \otimes {I_{n}}) \Psi^{-1}$ 
and \eqref{eq:Psi-inv}.
\end{remark}

\begin{proposition}\label{prop:feedforward-tracking}
Under Assumption \ref{assum:tracking-prob} and Lemma \ref{lemma:tracking-steady-state}, 
the output $y(t)$ tracks reference $r(t)$ with no steady-state error 
via dynamic tracking compensator \eqref{eq:dynamic-tracking-comp}, 
where the tracking control input realizes minimum attention.
\end{proposition}

\section{Numerical simulations}\label{sec:simulations}

In this section, we perform several numerical examples to illustrate the effectiveness of the designed sparse feedback control, which gives a closed-loop optimal solution by using a dynamic linear compensator.
\subsection{Single input}
At first, we consider a single-input control system modeled as a linearized cart-pole system, and the parameters benchmark are similar to \citep[Sec.~VI]{blanchini2003relatively}, in which the mass of the cart is $0.29~[\mathrm{kg}]$, mass of the pole is $0.1~[\mathrm{kg}]$, length of the pole is $1~[\mathrm{m}]$, gravity acceleration is $9.81~[\mathrm{m}/\mathrm{s}^{2}]$, and friction is neglected.

We then execute time-discretization of this continuous system using zero-order-hold (ZoH) with sampling $\Delta{t}=0.3~[\mathrm{s}]$, then the system matrices of form \eqref{eq:DT-LTI} are given by
\begin{align*}
    A = 
    \begin{bmatrix}
        1 & 0.3 & 0.1377 & 0.0143\\
        0 & 1 & 0.8256 & 0.1377\\
        0 & 0 & 0.4628 & 0.2441 \\
        0 & 0 & -3.2198 & 0.4628 
    \end{bmatrix},
    \quad
    B=
    \begin{bmatrix}
    0.1514\\
    0.9850\\
    -0.1404\\
    -0.8416
    \end{bmatrix}.
\end{align*}
Meanwhile, the output matrices with respect to state-input constraints \eqref{constrs:state-input} are set to
\begin{align*}
     C = \begin{bmatrix}
       0 & 0 & 1 & 0\\
       0 & 0 & 0 & 0
   \end{bmatrix},
   \quad \quad
   D =
   \begin{bmatrix}
       0\\
       1
   \end{bmatrix},
\end{align*}
then we have $y^{\top}(t)=\big[x_{3}(t)~u(t)\big]$, which means that the enforced constrained state is only third component of the state $x_{3}(t)$ and the imposed constrained input is $u(t)$.
By selecting the suitable variable $s$, it gives state and input constraints as follows
\begin{align*}
    |x_{3}(t)|\leq 1,\quad \quad |u(t)|\leq 1.
\end{align*}

Next, we simulate the discrete-time controlled system, the target is to drive the cart state from a non-zero initial state $x_{0}^{\top}=[0.9453,~0.7465,~0.7506,~0.4026]$ to the zero terminal state $x(N)=0$ in a finite $N=40$ steps (i.e, taking $\textrm{Tfinal}=({5*N*\Delta{t}})/{4}=15~[\mathrm{s}]$ in a state-space representation).
In order to realize the sparse feedback controller, we need to solve Problems ~\ref{prob:sparse-matrix-optim} and \ref{prob:sparse-feedback-realization} to seek the closed-loop $\ell_{1}$ optimal solution. 
By computing, the total CPU time in PC is $0.38~[\mathrm{s}]$ in MATLAB R2020b using \texttt{CVX} \cite{grant2008cvx}, and the found optimal value $\|{U}^{\ast}\|_{1}=6.4204$. 

\begin{figure}[!t]
\begin{center}
\includegraphics[width=0.6\linewidth]{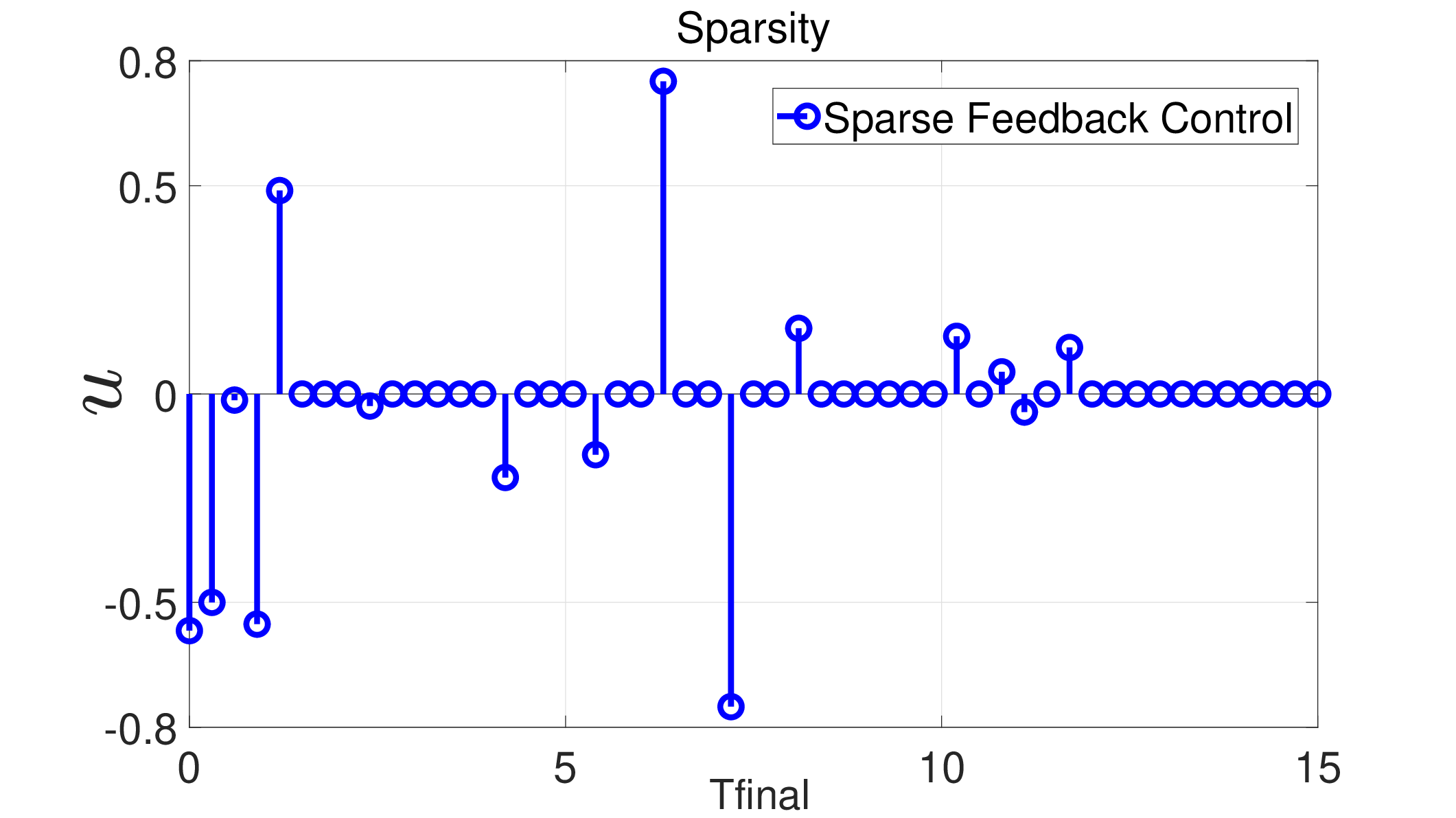}
\includegraphics[width=0.6\linewidth]{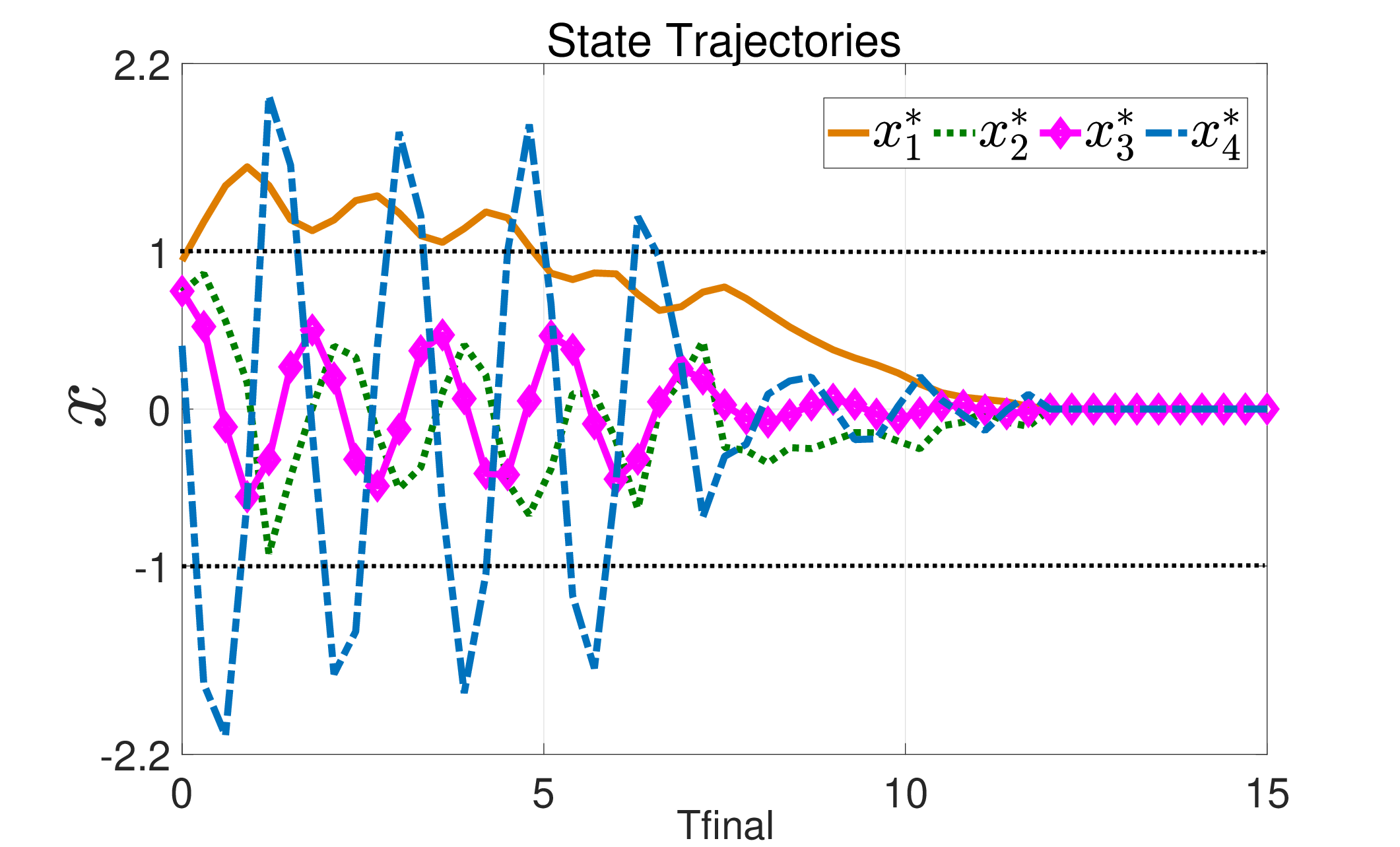}    
\caption{Single input case: optimal sparse feedback control $u^{\ast}(t)$ (top) and the regarding optimal state trajectories $x^{\ast}(t)$ (bottom). The black dot line represents the constraint on the pole angle $|x_{3}|\leq1$.} 
\label{fig:SISO-Input-state}
\end{center}
\end{figure}
Figure~\ref{fig:SISO-Input-state} illustrates the related closed-loop optimal input and state trajectories, respectively. 
It reflects the input sparsity on sparse feedback control, in which the control sequence is with less active components, and the optimal control meets constraint $|u(t)|\leq1$. 
In addition, it plots the optimal state trajectories, where the pole angle $x_{3}$ fluctuates between the bounds $-1$ and $1$, satisfying the prescribed state constraint $|x_{3}(t)|\leq1$. Meanwhile, the trajectories of four different states start from an initial state $x_{0}$ and eventually converge to zero state as time tends to a fixed steps under the dynamic compensator \eqref{eq:dynamic-state-feedback}, this implies that the closed-loop stabilization is achieved. 

\begin{figure}[!t]
\begin{center}
\includegraphics[width=0.6\linewidth]{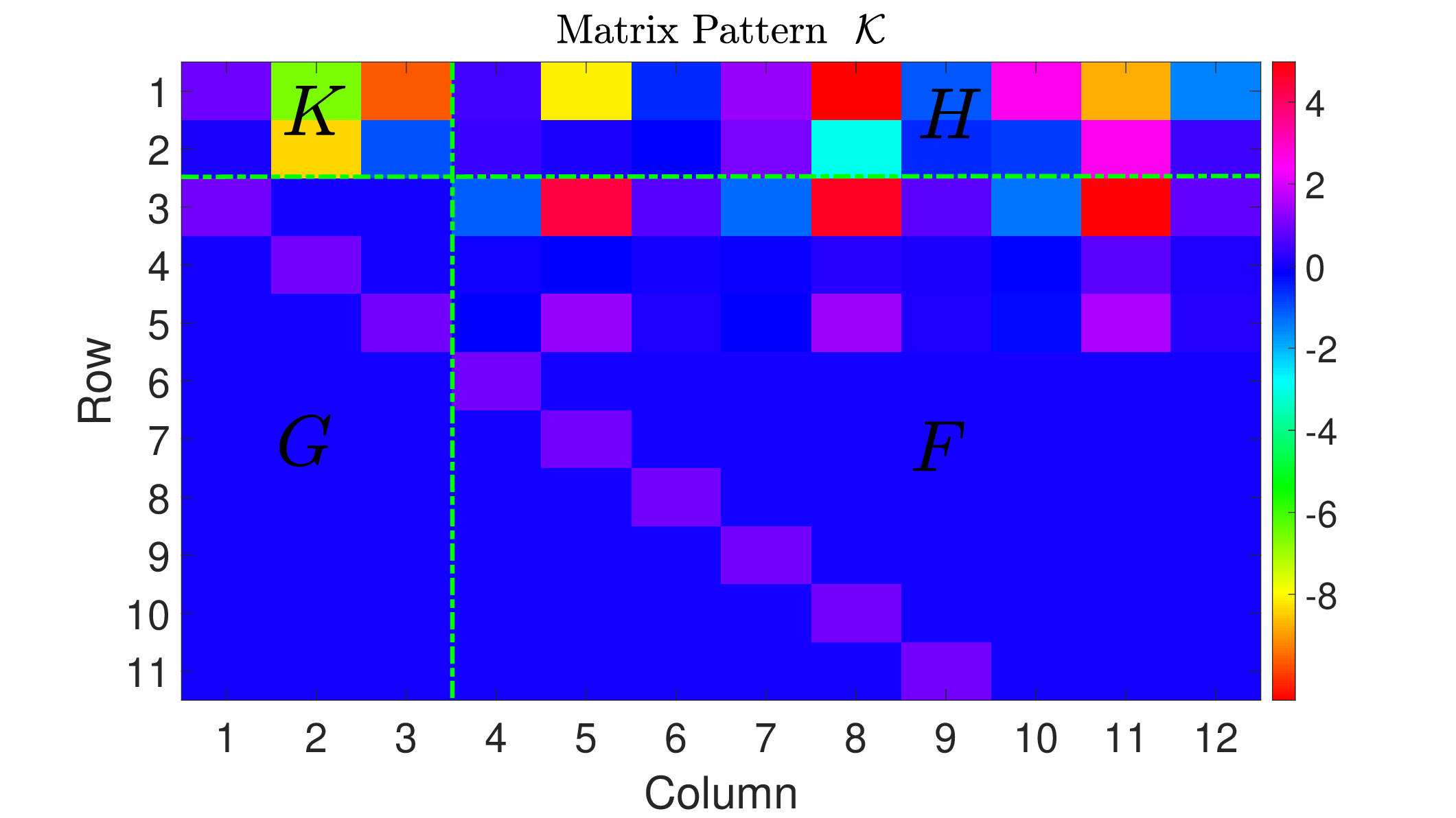}    
\caption{The pattern of dynamic linear Compensator $\mathcal{K}$.} 
\label{fig:pattern}
\end{center}
\end{figure}

\begin{figure}[!t]
\begin{center}
\includegraphics[width=0.6\linewidth]{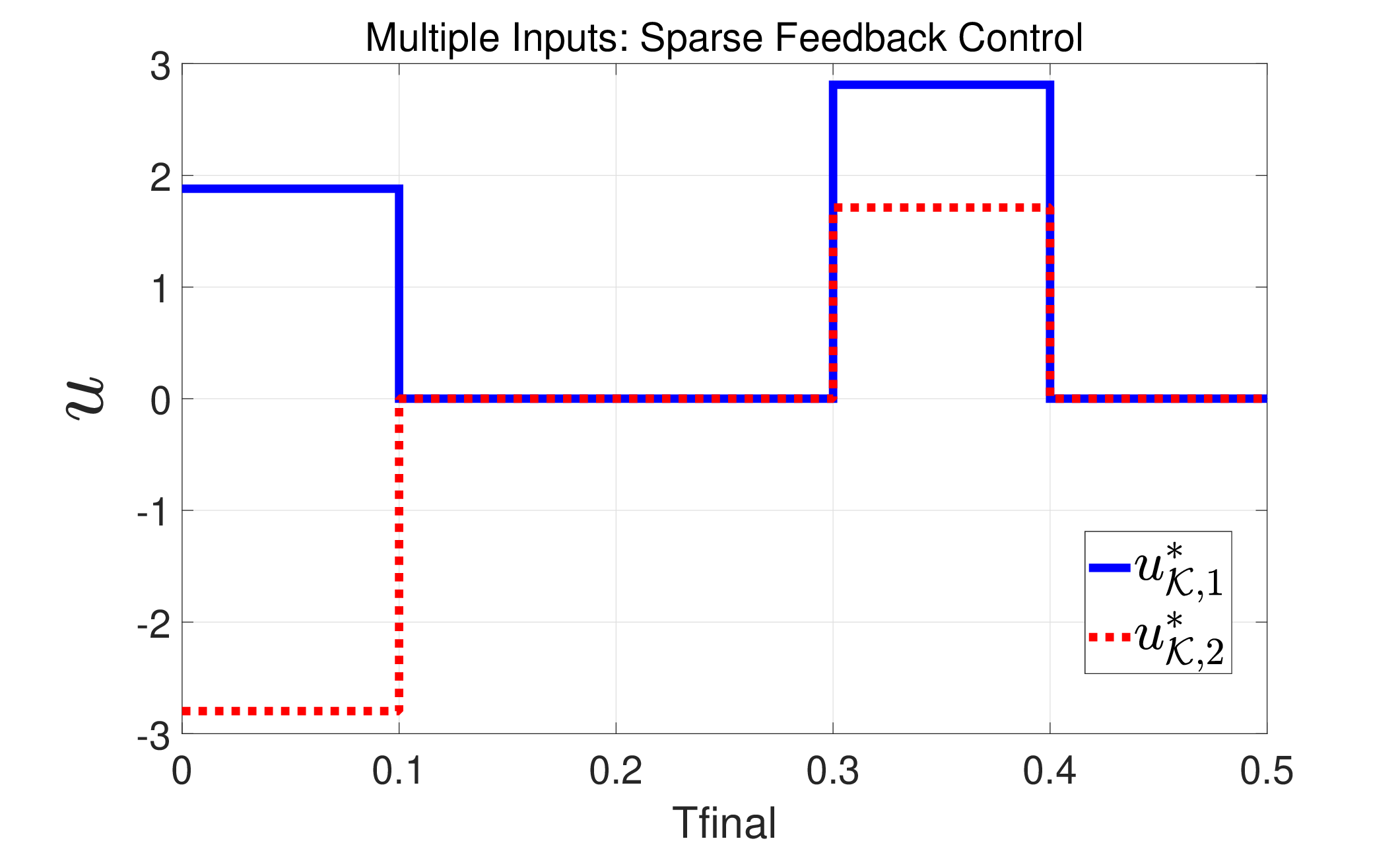} 
\includegraphics[width=0.6\linewidth]{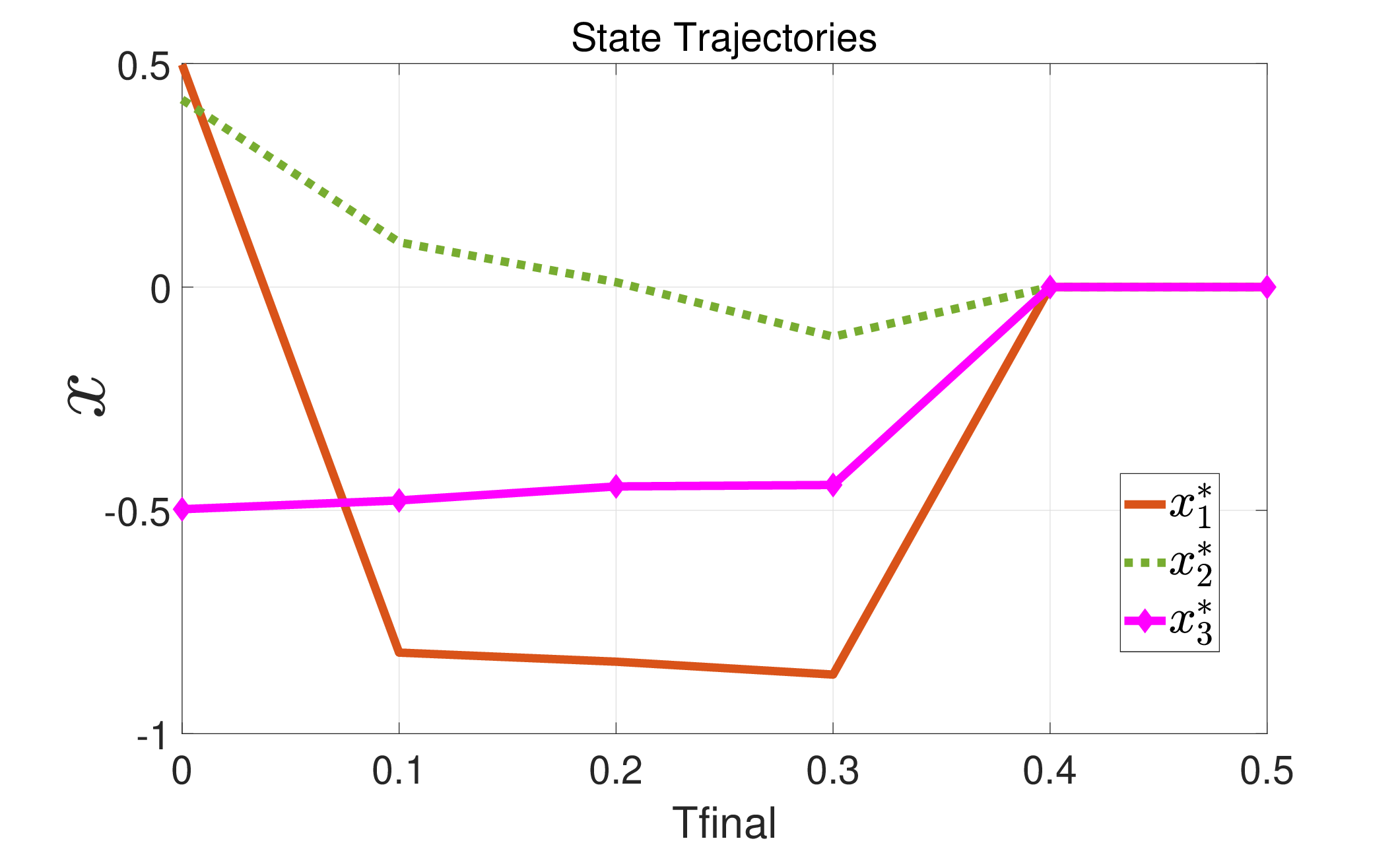}    
\caption{Multiple control inputs case: optimal sparse feedback control inputs $u_{\mathcal{K}}^{\ast}(t)$ (top) and the corresponding optimal state trajectories $x^{\ast}(t)$ (bottom).} 
\label{fig:Multi-inputs-state}
\end{center}
\end{figure}

\subsection{Multiple inputs numerical benchmark}
As a second numerical simulation, we show the result that the synthesized dynamic controller \eqref{eq:dynamic-state-feedback} is useful for sparse feedback control of multi-input control system. We here consider a discretized version of third-order system with two control inputs. Using a ZoH sampling time of $\Delta{t}=0.1~[s]$, the discrete system matrices are given by
\begin{align*}
    A
    =
    \begin{bmatrix}
     1.1133  &  0.0177 &  -0.1478 \\
    0.0177  &  1.4517  &  0.2514 \\
    0.0418  &   0.2758  &  0.9208
    \end{bmatrix},
    \quad
    B=
    \begin{bmatrix}
    0.0031  &  0.5218\\
    0.0121  &  0.1486\\
    0.0957  &  0.1202\\
    \end{bmatrix},
\end{align*}
and output matrices is chosen as $C=0$ and $D=I$, which means that only the restriction on the control inputs $|u_{i}(t)|\leq 10$ by taking $s_{i}=10$, $\forall{i=1,2}$.
We now randomly generate initial data $x_{0}\in[-1,1]^{3}$ and each input channel of time length is $N=4$, then the related final time in state-space is as $\text{Tfinal}=0.5~[s]$. 

After the state-input matrices ${X}$ and ${U}$ were calculated (see \eqref{eq:sec5-MIMO-data} in Appendix),
we obtained the optimal value $\|U^{\ast}\|_{1}=24.1544$. Due to the fact that the augmented matrix $\Psi$ is invertible \eqref{eq:Psi-inv}, then the controller $\mathcal{K}$ with real matrices $(K,H,G,F)$ is computationally efficient. Figure~\ref{fig:pattern} reveals the pattern of the compensator $\mathcal{K}$, in which the color-bar reports the level of real values of the correlation elements in matrix $\mathcal{K}$. Theorem \ref{them:sparse-feedback-realziation} implies that we require the knowledge of the matrices $H$ and $K$  (see \eqref{eq:sec5-MIMO-KH-data} in Appendix) to synthesize the sparse feedback control, as follows
\begin{align*}
    u_{\mathcal{K}}^{\ast}=
\begin{bmatrix}
    1.8798 & 0.0000 & -0.0000 & 2.8111 & 0.0000 & 0.0000\\
   -2.7970 & 0.0000 & -0.0000 &  1.7122 & -0.0000 &  0.0000
\end{bmatrix}.
\end{align*}

As shown in Figure~\ref{fig:Multi-inputs-state}, the optimal feedback control signals contain two components, where both control inputs are along the input constraints $|u_{\mathcal{K},i}(t)|\leq10$, $\forall{i}=1,2$. 
Clearly, the inferred feedback control sequences are sparse as desired.
From this figure, it appears that the optimal state trajectories converge to zeros with minimum control effort.

\subsection{Tracking problem}
Finally, we show a numerical example to illustrate the effectiveness of our extended dynamic tracking compensator \eqref{eq:dynamic-tracking-comp} for tracking problem, in Section~\ref{sec:sfb-tracking}. By taking a continuous second-order harmonic oscillator as
\begin{align*}
\begin{aligned}
       \dot{x}(t) &=
    \begin{bmatrix}
        0 & 1\\
        1 & 0
    \end{bmatrix}x(t)
    + \begin{bmatrix}
        -2\\
        1
    \end{bmatrix}u(t),
    \quad 
    x(0)=
     \begin{bmatrix}
        0.5\\
      -0.5
    \end{bmatrix},\\
    y(t) &= 
    \begin{bmatrix}
        0 & 1
    \end{bmatrix}x(t),
\end{aligned}
\end{align*}
and discretized it with a ZoH sampling with time period $\Delta{t}=0.2~[s]$. The time horizon is $N=5$.
For convenience, we model a step reference signal \eqref{eq:reference-dynamics} as $r_{+}=1$ for all $t\geq{0}$. In this case, the default value of steady-state of reference is $r_{+}=r_{0}=1$. 

\begin{figure}[!t]
\begin{center}
\includegraphics[width=0.6\linewidth]{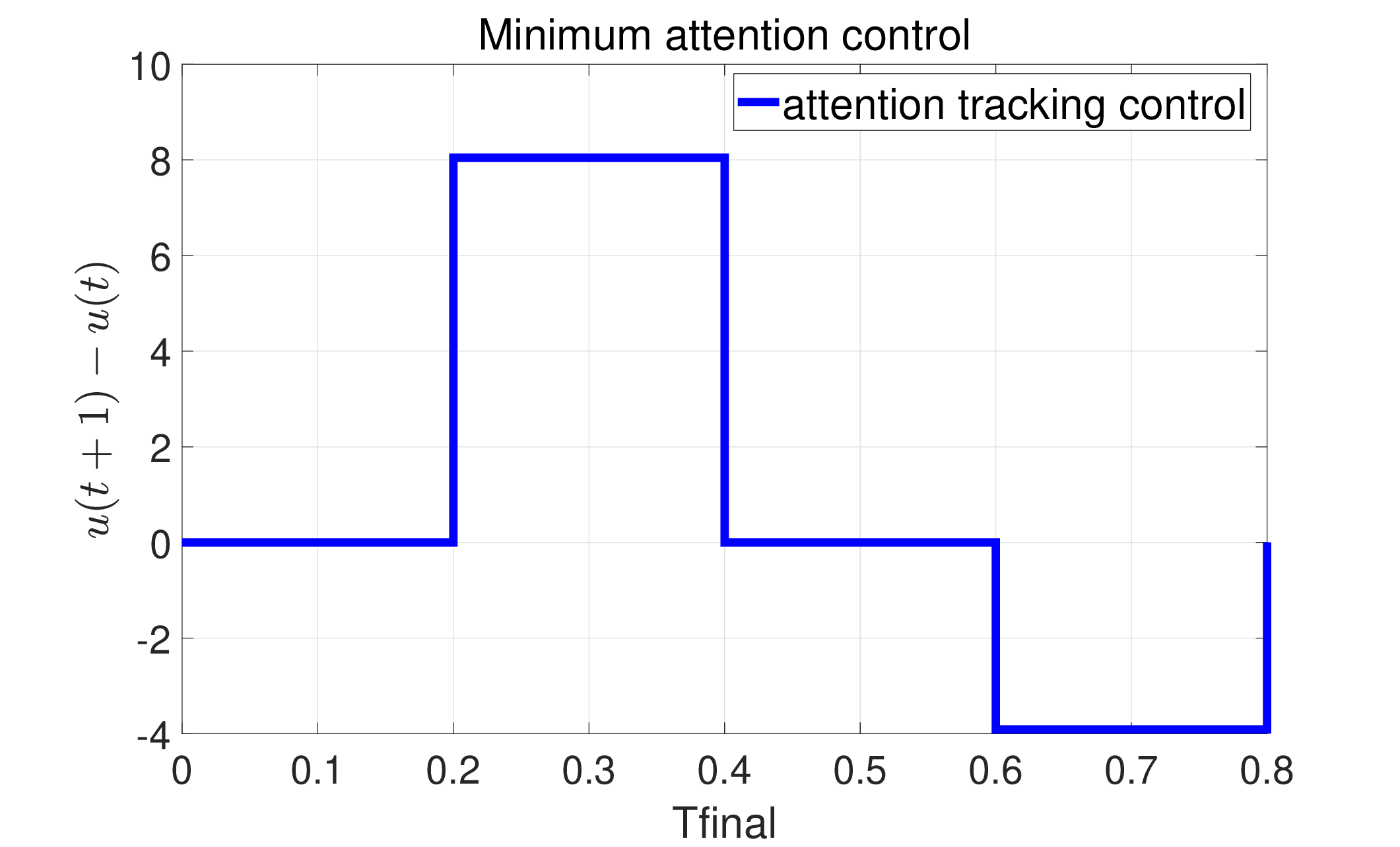}    
\includegraphics[width=0.6\linewidth]{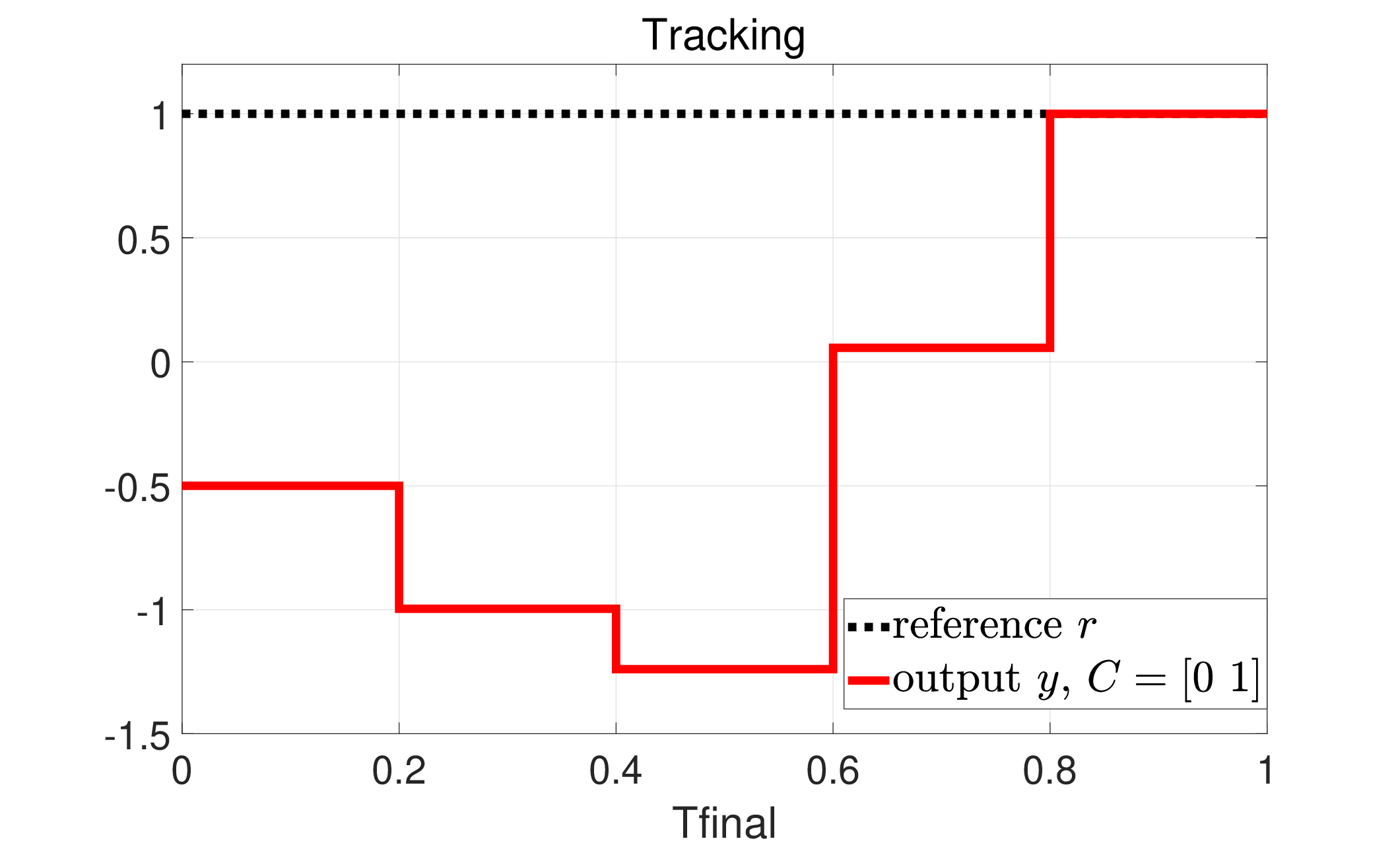}
\caption{Minimum attention tracking control (i.e., the difference of control input) (top) and the tracking trajectories of performance $y$ with respect to a step reference $r$ (bottom).} 
\label{fig:track-problem}
\end{center}
\end{figure}
Focusing on the solution $(X,U)$ determined by Problem \ref{prob:track-attention}, we then make use of the feedback realization technique to derive the feedback gain matrices $(K,H,G,F)$, where the relevant numerical results are shown in \eqref{eq:sec5-tracking-data-XUKHGF} in Appendix.

The problem at hand is to achieve tracking, we found that the determinants $\mathrm{det}(I-(A+BK))=\mathrm{det}(I-X_{1})=0.2475\neq{0}$ and $\mathrm{det}({I-F})=2.4902\neq{0}$, fitting two conditions in Lemma \ref{lemma:tracking-steady-state}, and the related feedforward gains can be selected as 
\begin{align*}
    M = -3.0837,
    \quad 
    L=
    \begin{bmatrix}
        0.5 & -1 & 0 & 0 & 0 & 0 
    \end{bmatrix}^{\top}.
\end{align*}
Based on the above analysis, the implemented dynamic tracking compensator 
\begin{align*}
    u_{\mathcal{K}_{r}}^{\ast}=
    \begin{bmatrix}
        -3.6367 &  -3.6367  &  4.4087  &  4.4087  &  0.5000  &  0.5000
    \end{bmatrix}
\end{align*}
can be easily applied to the discrete LTI plant so as to achieve  tracking.
Figure~\ref{fig:track-problem} depicts the evolution of minimum attention tracking control (top) and the corresponding tracking trajectories. It can be see that the performance output signal $y(t)$ gradually tracks a step reference signal $r(t)=1$ under the specified
time steps.

\section{Conclusion}\label{sec:conclusion}
In this paper, we have proposed a sparse feedback controller from open-loop solution to closed-loop realization. By means of implementing a dynamic linear compensator, the stability, optimality, and sparsity of the closed-loop $\ell_{1}$ optimal control are ensured. Besides, we extended the result to a tracking problem to achieve the minimum attention control. Finally, the simulations illustrated the effectiveness of the proposed sparse feedback control. 
Future work will focus on investigating sparse feedback control for model-free 
systems, exploiting input-state/output data to design a data-driven sparse controller. 
Additionally, an intriguing aspect to explore is sparse output feedback control. 
However, directly replicating the current results may be challenging 
due to the influence of generic initial conditions and the initialization of 
the dynamic output compensator, which are associated with the model matching problem.

\bibliographystyle{tfnlm}
\bibliography{arXiv-SFB07}

\begin{landscape}
\appendix
\section{Numerical results in Section~\ref{sec:simulations}}
\begin{flalign}
\begin{aligned}
    {X}&=
\left[ {\begin{array}{cccccccccccc}
1.0000 &  0 & 0 & 1.1753 &-4.2597 & -0.7724 & 1.2898 & -4.5326 & -0.8472 & 1.4109 & -4.8391 & -0.9265 \\
0 &  1.0000 & 0 & 0.0443 & 0.1526 & -0.0283 & 0.1183 & -0.2060 & -0.0772 & 0.2406 & -0.7377 & -0.1578 \\
0 &  0 & 1.0000 & 0.1322 &-1.4008 & -0.0892 & 0.1830 & -1.4257 & -0.1222 & 0.2551 & -1.5589 & -0.1692
\end{array} } \right],  
\\
{U}&= 
\left[ {\begin{array}{cccccccccccc}
 0.8012 & -7.2752 & -9.1168 & 0.0000  & -0.0000 & -0.0000 & -0.0000 & 0.0000 & 0.0000 & -0.0626 & 6.8408 & 0.0578\\
0.1140  & -8.1530 & -1.1419 & -0.0000 &  0.0000 &  0.0000 & -0.0000 & 0.0000 & 0.0000 & -2.9457 &  9.8661 &  1.9337
\end{array} } \right].
\end{aligned}
\label{eq:sec5-MIMO-data}
\tag{A.1}
\end{flalign}
\begin{flalign}
\begin{aligned}
K&=
\begin{bmatrix}
0.8012  & -7.2752  & -9.1168\\
0.1140  & -8.1530  & -1.1419
\end{bmatrix},\\
H &= 
{\begin{bmatrix}
 0.5857  & -8.2481 &  -0.4003  &  1.4963 & -10.8652  & -0.9970   & 2.8829  & -8.8614 &  -1.8905\\
0.3782  &  0.1298  & -0.2447   & 1.0268  & -2.7911  & -0.6724   &-0.8534   & 2.6232  &  0.5596
\end{bmatrix}.}
\end{aligned}
\label{eq:sec5-MIMO-KH-data}
\tag{A.2}
\end{flalign}
\begin{align}
\begin{aligned}
    X & =
    \begin{bmatrix}
    1.0000 &   0    &  -0.8758 &  -2.1177 & -2.5880 & -4.0772 & -1.2121 & -1.9096\\
     0     & 1.0000 &  1.0001  &   1.9972 &  1.6427 & 2.5880  & 0.6924  & 1.0908\\
    \end{bmatrix},\\
    U & =
    \begin{bmatrix}
    4.9553 & 6.0613 & 4.9553 & 6.0613 & -2.8673 & -4.5173 & -2.8673 & -4.5173
    \end{bmatrix},\\
    K &= \begin{bmatrix}  4.9553 &   6.0613\end{bmatrix},\\
    H & = 
    \begin{bmatrix}
     3.2331 && 4.4499 & 0.0000 & 0.0000 & -1.0576 & -1.6662
    \end{bmatrix}\\
    G & =
    \begin{bmatrix}
      1  &   0  &   0  &   0   &  0   &  0\\
     0   &  1   &  0   &  0   &  0   &  0
    \end{bmatrix}^{\top},\\
    F & =
    \begin{bmatrix}
    0.8758 &  2.1177  &  2.5880  &   4.0772 &  1.2121  &  1.9096 \\
   -1.0001 &  -1.9972 & -1.6427  &  -2.5880 &  -0.6924 &  -1.0908\\
    1.0000 &        0 &       0  &       0  & 0        &    0    \\
         0 &  1.0000  &      0   &      0   &      0   &    0    \\
         0 &        0 &   1.0000 &        0 &        0 &    0    \\
         0 &        0 &       0  &  1.0000  &      0   &     0   \\
    \end{bmatrix}.
    \end{aligned}
        \label{eq:sec5-tracking-data-XUKHGF}
        \tag{A.3}
\end{align}
\end{landscape}

\end{document}